\documentclass[a4paper]{article}
\usepackage[english]{babel}
\usepackage[english]{layout}
\usepackage[title]{appendix}
\usepackage[latin1]{inputenc}

\usepackage{amsmath,amsthm,amssymb,bm,bbm}

\usepackage{fancyhdr} 
\numberwithin{equation}{section}

\newtheorem{definition}{Definition}[section]

\newtheorem{lemma}{Lemma}[section]
\newtheorem{theorem}{Theorem}[section]
\newtheorem{corollary}{Corollary}[section]

\newcommand{\bbC}{\mathbb{C}}

\newcommand{\betheop}{\hat{\beta}}

\renewcommand{\l}{\lambda}
\newcommand{\bl}{\boldsymbol{\lambda}} 
\renewcommand{\L}{\Lambda}
\newcommand{\bmu}{\boldsymbol{\mu}} 
\renewcommand{\o}{\omega}
\renewcommand{\O}{\Omega}

\newcommand{\pop}{\hat{p}}

\newcommand{\tr}{\operatorname{Tr}}
\newcommand{\cH}{\ensuremath{\mathcal{H}}}

\newcommand{\ox}{\otimes}
\newcommand{\LI}{\ensuremath{\underset{1}{L}}}
\newcommand{\LII}{\ensuremath{\underset{2}{L}}}
\newcommand{\xp}{X^{+}}
\newcommand{\xm}{X^{-}}

\begin{document}

\title{\LARGE \textbf{Algebraic Bethe Ansatz for deformed Gaudin model}}

\author{\textsf{N. ~Cirilo ~Ant\'onio,}
\thanks{E-mail address: nantonio@math.ist.utl.pt}
\textsf{ ~~N. ~Manojlovi\'c}
\thanks{E-mail address: nmanoj@ualg.pt}
\textsf{ and A. ~Stolin}
\thanks{E-mail address: 
astolin@chalmers.se} \\
\\
\textit{$^{\ast}$Centro de An\'alise Funcional e Aplica\c{c}\~oes}\\
\textit{Departamento de Matemática, Instituto Superior Técnico} \\
\textit{Av. Rovisco Pais, 1049-001 Lisboa, Portugal} \\
\\
\textit{$^{\dag}$Grupo de F\'{\i}sica Matem\'atica da Universidade de Lisboa} \\
\textit{Av. Prof. Gama Pinto 2, PT-1649-003 Lisboa, Portugal} \\
\\
\textit{$^{\dag}$Departamento de Matem\'atica, F. C. T.,
Universidade do Algarve}\\ 
\textit{Campus de Gambelas, PT-8005-139 Faro, Portugal}\\
\\
\textit{$^{\ddag}$Departament of Mathematical Sciences}\\
\textit{University of G\"oteborg, SE-412 96 G\"oteborg, Sweden}}
\date{}


\maketitle
\thispagestyle{empty}
\vspace{10mm}
\begin{abstract}
The Gaudin model based on the $sl_{2}$-invariant r-matrix with an extra Jordanian term depending on the spectral parameters is considered. The appropriate creation operators defining the Bethe states of the system are constructed through a recurrence relation. The commutation relations between the generating function $t(\l)$ of the integrals of motion and the creation operators are calculated and therefore the algebraic Bethe Ansatz is fully implemented. The energy spectrum as well as the corresponding Bethe equations of the system coincide with the ones of the $sl_{2}$-invariant Gaudin model. As opposed to the $sl_{2}$-invariant case, the operator $t(\l)$ and the Gaudin Hamiltonians are not Hermitian. Finally, the inner products and norms of the Bethe states are studied.
\end{abstract}

\newpage
\section{Introduction}

A model of interacting spins in a chain was first considered by Gaudin \cite{Gaudin,Gaudin2}. In his approach these models were introduced as a quasi-classical limit of the integrable quantum chains. Moreover, the Gaudin models were extended to any simple Lie algebra, with arbitrary irreducible representation at each site of the chain 
\cite{Gaudin2}.

In the framework of the Quantum Inverse Scattering Method (QISM) \cite{SkTF}\nocite{KS82,STS,KBI}--\cite{Faddeev} the Gaudin models are based on classical r-matrices which have a unitarity property 
\begin{equation}
\label{r-unitarity}
r_{12}(\mu,\nu) = -r_{21}(\nu,\mu).
\end{equation}
Then the Gaudin Hamiltonians are related to classical r-matrices  
\begin{equation}
\label{Gaudinham}
H^{(a)} = \sum_{b\neq a}^{N} r_{ab}(z_{a},z_{b}).
\end{equation}
The condition of their commutativity $[H^{(a)},H^{(b)}]=0$ is nothing else but the classical Yang-Baxter equation 
\begin{equation}
\label{cYBE}
[ r_{ab}(z_{a},z_{b}), r_{ac}(z_{a},z_{c}) + r_{bc}(z_{b},z_{c})] + [r_{ac}(z_{a},z_{c}),r_{bc}(z_{b},z_{c})]=0
\end{equation}
where $r$ belongs to the tensor product $\mathfrak{g} \otimes \mathfrak{g}$ of a Lie algebra $\mathfrak{g}$, or its
representations, and the indices fix the corresponding factors in the $N$-fold tensor product. 

Considerable attention has been devoted to Gaudin models based on the classical r-matrices of simple Lie algebras \cite{Sklyanin89}\nocite{Sklyanin99,BF,FFR94}--\cite{RV95}  and Lie superalgebras \cite{KM2, KM3, LSU01}. The spectrum and corresponding eigenfunctions were found using different methods such as coordinate and algebraic Bethe ansatz, separated variables, etc. Correlation functions for Gaudin models were explicitly calculated as combinatorial expressions obtained from the Bethe ansatz \cite{Gaudin2}. In the case of the $sl_2$ Gaudin system, its relation to Knizhnik-Zamolodchikov equation of conformal field theory \cite{RV95} or the method of Gauss factorization \cite{Sklyanin99}, provided alternative approaches to computation of correlations. The non-unitary r-matrices and the corresponding Gaudin models have been studied recently, see \cite{Skrypnyk} and the references therein.


A classification of solutions to classical Yang-Baxter equation \eqref{cYBE}, with the property \eqref{r-unitarity}, has been done in \cite{BelavinDrinfeld}. In particular, a class of rational solutions corresponding to the simple Lie algebras $sl_{n}$ was discussed in \cite{BelavinDrinfeld} and then further studied and generalized in \cite{Stolin1,Stolin}. In the $sl_{2}$ case, up to gauge equivalence, only two such solutions exist: the $sl_{2}$-invariant r-matrix with an extra Jordanian term (deformation) 
\begin{equation}
\label{rJordan}
r^{J}_{\xi}(\mu,\nu) = \frac{c_{2}^{\otimes}}{\mu-\nu} +\xi (h\otimes \xp - \xp\otimes h),
\end{equation}
where $c_{2}^{\ox}= h \otimes h + 2 \left( X^+ \otimes X ^- +  X^-\otimes X ^+\right)$ denotes the quadratic tensor Casimir of $sl_{2}$, and the $sl_{2}$-invariant r-matrix with a deformation depending on the spectral parameters 
\begin{equation}
\label{defJordan}
r_{\xi}(\mu,\nu) = \frac{c_{2}^{\otimes}}{\mu-\nu} +\xi \left(h\otimes (\nu \xp) - (\mu \xp)\otimes h\right).
\end{equation}

The Gaudin model based on the r-matrix \eqref{rJordan} was studied as the quasi-classical limit of the corresponding quantum spin chain \cite{KSto, KST}, which is a deformation of the XXX spin chain where the Jordanian twist  is applied to the Yang R-matrix \cite{Kulish2000,LMO}. The algebraic Bethe Ansatz for this Gaudin model was fully implemented in \cite{CAM1}, following ideas given in \cite{KM2, KM3, Kulish2002}.

In the framework of the QISM, the classical Yang-Baxter equation guarantees for the Gaudin model the existence of integrals of motion in involution. A more difficult problem is to determine the spectra and the corresponding eigenvectors. The Algebraic Bethe Ansatz (ABA) is a powerful tool which yields the spectrum and corresponding eigenstates in the applications of the QISM to the models for which highest weight type representations are relevant, like for example Gaudin models, quantum spin systems, etc \cite{TakhtajanFaddeev, Faddeev, SklyaninTakebe}. However the ABA has to be applied on the case by case basis.

An initial step in the application of the Algebraic Bethe Ansatz to the Gaudin model based on the r-matrix \eqref{defJordan}, is to define appropriate creation operators which yield the Bethe states and the spectrum of the corresponding Gaudin Hamiltonians. 
The operators $B^{(k)}_{M}(\mu_{1},\ldots,\mu_{M})$ used for this model, introduced in \cite{KMSS}, are defined through a recursive relation that generalizes the one used in \cite{CAM1} for the Gaudin model based on the r-matrix \eqref{rJordan}. The commutation relations of the $B^{(k)}_{M}(\mu_{1},\ldots,\mu_{M})$ operators with the entries of the L-operator are straightforward to calculate, see lemma \ref{lem:2}. The main step, lemma \ref{lem:4} below,  in the implementation of the ABA is to calculate explicitly the commutation relations between the creation operators $B_{M}(\mu_{1},\ldots,\mu_{M})$ and the generating function of integrals of motion $t(\l)$ of this Gaudin model. 
Finally, it is confirmed that the operator $t(\l)$ is not Hermitian with respect to the natural inner product.

Besides determining the spectrum of the system, via ABA, another important problem is to calculate the inner product of the Bethe vectors. To this end it is necessary to introduce the duals of the creation operators $B^{\ast}_{M}(\mu_{1},\ldots,\mu_{M})$. 
A method developed in \cite{CAM1} to determine the inner product of the Bethe vectors will be used here. The main idea is to establish a relation between the creation operators of the system and the ones in the $sl_2$-invariant case, \textit{i.e.} when $\xi=0$. The advantage of this approach is that the cumbersome calculation involving strings of commutators between entries of the L-operator becomes unnecessary. As opposed to the $sl_{2}$-invariant case, the Bethe vectors $\Psi_{M}(\mu_{1},\ldots,\mu_{M})$ for the deformed Gaudin model are not orthogonal for different $M$'s, with respect to the natural inner product considered.

The article is organized as follows. In section 2 we study the deformed Gaudin model based on the r-matrix \eqref{defJordan}, emphasizing the operators $B^{(k)}_{M}$ and their properties. Using 
the commutation relations of the creation operators $B_{M}$ with the  generating function of integrals of motion $t(\l)$, in section 3 the Bethe vectors of the system are given. It is also confirmed that the spectrum is the one of the invariant system. The duals of the creation operators are introduced in section 4 in order to obtain the expressions for the inner products and the norms of Bethe states. The results of the study are briefly summarized in the conclusions. The proofs of the lemmas are given in the appendix.

\section{Algebraic Bethe Ansatz for deformed Gaudin model}

The Gaudin model based on the $sl_{2}$-invariant r-matrix with an extra Jordanian term depending on the spectral parameters  
\begin{equation}
\label{r-matrix}
\begin{split}
r_{\xi}(\l,\mu) = \frac{c_{2}^{\ox}}{\l-\mu} + \xi p(\l,\mu)&= \frac{1}{\l-\mu}\bigg (h\ox h + 2(\xp\ox\xm+\xm\ox\xp)\bigg)\\
&+\xi \bigg (h\ox (\mu\xp) - (\l\xp)\ox h\bigg), \end{split}
\end{equation}
where $(h,X^{\pm})$ are the standard $sl_2$ generators
\begin{equation}
\label{sl2}
[h,X^{+}]=2 X^{+}, \quad  [X^{+},X^{-}]=h,  \quad [h,X^{-}]=-2 X^{-},
\end{equation}
is considered. The matrix form of $r_{\xi}(\l,\mu)$ in the fundamental representation of $sl_{2}$ is given explicitly by \cite{KMSS}
\begin{equation}
\label{r-m:mf}
r_{\xi}(\l,\mu)=
\begin{pmatrix}
 \displaystyle\frac{1}{\l -\mu } & \mu\xi  & -\l\xi  & 0 \\[0.25cm]
 0 & \displaystyle-\frac{1}{\l -\mu } & \displaystyle\frac{2}{\l -\mu } & \l\xi  \\[0.25cm]
 0 & \displaystyle\frac{2}{\l -\mu } &\displaystyle -\frac{1}{\l -\mu } & -\mu  \xi  \\[0.25cm]
 0 & 0 & 0 & \displaystyle\frac{1}{\l -\mu } 
\end{pmatrix},
\end{equation}
here $\l,\mu\in\bbC$ are the so-called spectral parameters and $\xi\in\bbC$ is a deformation parameter. The next step is to introduce the L-operator of the Gaudin model \cite{Sklyanin89}. In the present case the entries of the L-operator
\begin{equation}
\label{L-mat}
L(\l) = 
\begin{pmatrix}
h(\l) & 2X^{-}(\l) \\[0.1cm]
2X^{+}(\l) & - h(\l)
\end{pmatrix}
\end{equation}
are given by
$$h(\l) = \sum_{a=1}^{N} \left(\frac{h_{a}}{\l - z_a} + \xi z_{a}\xp_{a}\right),$$
\begin{equation}
\label{eq:09}
\xm(\l) = \sum_{a=1}^{N} \left(\frac{\xm_{a}}{\l - z_a} - \frac{\xi}{2}\l h_{a}\right), \
\xp(\l) = \sum_{a=1}^{N} \frac{\xp_{a}}{\l - z_a},
\end{equation}
where $h_{a} = \pi_{a}^{(\ell_{a})} (h) \in \operatorname{End}(V_{a}^{(\ell_{a})})$, $X^{\pm}_{a} = \pi_{a}^{(\ell_{a})} (X^{\pm}) \in \operatorname{End}(V_{a}^{(\ell_{a})})$ and $\pi_{a}^{(\ell_{a})}$ is an irreducible representation of $sl_{2}$ whose representation space is $V_{a}^{(\ell_{a})}$ corresponding to the highest weight $\ell_{a}$ and the highest weight vector $\o_{a}\in V_{a}^{(\ell_{a})}$, i.e. $\xp_{a}\o_{a}=0$ and $h_{a}\o_{a}=\ell_{a}\,\o_{a}$, at each site $a=1,\ldots,N$. Notice that $\ell_{a}$ is a nonnegative integer and the $(\ell_{a} + 1)$-dimensional representation space $V_{a}^{(\ell_{a})}$ has the natural Hermitian inner product such that (cf. \cite{Kosmann-Schwarzbach:2009fk}) 
$$(\xp_a)^{*}=\xm_a, \quad (\xm_a)^{*}=\xp_a \quad \text{and} \quad h_a^{*}=h_a.$$
The space of states of the system $\cH = V_{1}^{(\ell_{1})}\ox\cdots\ox V_{N}^{(\ell_{N})}$ is 
naturally equipped with the Hermitian inner product $\langle \cdot \lvert \cdot \rangle $ as a tensor product of the spaces $V_{a}^{(\ell_{a})}$ for $a=1,\ldots,N$.

The L-operator \eqref{L-mat} satisfies the so-called Sklyanin linear bracket \cite{Sklyanin89}
\begin{equation}
\label{S-l-bracket}
\left[\LI(\l),\LII(\mu)\right] = -\left[r_{\xi}(\l,\mu),\LI(\l) + \LII(\mu)\right].
\end{equation}
Both sides of this relation have the usual commutators of the $4 \times 4$ matrices  $\LI(\l) = L(\l) \otimes \mathbbm{1}$, $\LII(\l)= \mathbbm{1} \otimes L(\l)$ and $r_{\xi}(\l,\mu)$, where $\mathbbm{1}$ is the $2 \times 2$ identity matrix. The relation \eqref{S-l-bracket} is a compact matrix form of the following commutation relations 
\begin{align}
\left[h(\l),h(\mu)\right] &= 2\xi \left(\l\xp(\l) - \mu\xp(\mu)\right)\notag\\
\left[\xm(\l),\xm(\mu)\right] &= -\xi \left(\mu\xm(\l) - \l\xm(\mu) \right),\notag \\
\left[\xp(\l),\xm(\mu)\right] &= -\frac{h(\l) - h(\mu)}{\l - \mu} + \xi\mu\xp(\l),\notag\\ 
\left[\xp(\l),\xp(\mu)\right] &= 0, 
\label{loopalg}
\\
\left[h(\l),\xm(\mu)\right] &= 2\frac{\xm(\l) - \xm(\mu)}{\l - \mu} + \xi\mu h(\mu),\notag\\
\left[h(\l),\xp(\mu)\right] &= -2\frac{\xp(\l) - \xp(\mu)}{\l - \mu}.\notag
\end{align}
%

In order to define a dynamical system besides the algebra of observables a Hamiltonian should be specified. Due to the Sklyanin linear bracket \eqref{S-l-bracket} the generating function \cite{Sklyanin89}
\begin{equation}
\label{t(l)}
t(\l) = \frac{1}{2}\tr L^2(\l) = h^{2}(\l) -2 h^{\prime}(\l) + 2\left(2\xm(\l)+\xi\l\right)\xp(\l)
\end{equation}
satisfies
\begin{equation}
\label{eq:11}
t(\l)t(\mu) = t(\mu)t(\l).
\end{equation}
The pole expansion of the generating function $t(\l)$ is obtained by substituting \eqref{eq:09} into \eqref{t(l)}
\begin{equation}
\label{poleexp}
t(\l) = \sum_{a=1}^{N}\left(\frac{\ell_a(\ell_a +2)}{(\l - z_a)^{2}} + \frac{2H^{(a)}}{\l - z_a} \right) + 2\xi (1-h_{(gl)})\xp_{(gl)} + \xi^{2} \sum_{a,b=1}^{N}z_{a}z_{b}\xp_{a}\xp_{b}.
\end{equation}
The residues of the generating function $t(\l)$ at the points $\l=z_{a}$, $a=1,\ldots,N$ are the Gaudin Hamiltonians
\begin{equation}
\label{GHam}
H^{(a)} =  \sum_{b\neq a}^{N}\left(\frac{c_{2}(a,b)}{z_a - z_b} + \xi \left(z_{b}h_{a} X_{b}^{+} - z_{a}X_{a}^{+}h_{b} \right)\right),
\end{equation}
where $c_{2}(a,b) = h_{a} h_{b} + 2(\xp_{a}\xm_{b} + \xm_{a}\xp_{b})$. In the constant term of the expansion \eqref{poleexp} the notation
\begin{equation}
\label{globals}
Y_{(gl)}= \sum_{a=1}^{N}Y_{a},
\end{equation}
for $Y = (h,X^{\pm})$, was used to denote the generators of the so-called global $sl_2$ algebra. In the case when $\xi = 0$ the global $sl_2$ algebra is a symmetry of the system.

Moreover, from the formula \eqref{poleexp} it is straightforward to obtain 
\begin{equation}
\label{t-decomp}
t(\l) = t(\l)_{0} +2\xi\left(h(\l)_{0}\hat{X}^{+}_{(gl)} +\xp_{(gl)}- \l h_{(gl)}\xp(\l)\right) +\xi^{2} (\hat{X}^{+}_{(gl)})^{2},
\end{equation}
where $\hat{X}^{+}_{(gl)}=\sum_{a=1}^{N}z_{a}\xp_{a}$, $h(\l)_{0}=h(\l)|_{\xi=0}$ and $t(\l)_{0}=t(\l)|_{\xi=0}$ is the 
generating function of the integrals of motion in the $sl_{2}$-invariant case.

In order to obtain the Bethe states of the system it is necessary to fully implement the ABA. In the space of states $\cH$ 
the vector
\begin{equation}
\label{Omega}
\O_{+}=\o_{1}\ox\cdots\ox\o_{N}
\end{equation}
is such that $\langle \O_{+} \lvert \O_{+}\rangle =1$ and
\begin{equation}
\label{vacuum}
\xp(\l)\O_{+}=0, \quad h(\l)\O_{+}=\rho(\l)\O_{+},
\end{equation}
with
\begin{equation}
\label{Ro}
\rho(\l)=\sum_{a=1}^{N}\frac{\ell_{a}}{\l-z_{a}}.
\end{equation}

The commutators between the generating function \eqref{t(l)} and the entries of the L-operator \eqref{eq:09}
are
\begin{equation}
\label{eq:14}
\left[t(\l),\xp(\mu)\right] = 4\frac{h(\l)\xp(\mu) - h(\mu)\xp(\l)}{\l - \mu} - 4\,\xi\l\,\xp(\l)\xp(\mu),
\end{equation}

\begin{equation}
\label{eq:15}
\left[t(\l),h(\mu)\right] = 8\frac{\xm(\mu)\xp(\l) - \xm(\l)\xp(\mu)}{\l - \mu} - 4\,\xi(1+\mu\,h(\l))\xp(\mu),
\end{equation}

\begin{equation}
\label{eq:16}
\begin{split}
\left[t(\l),\xm(\mu)\right] &= 4\frac{\xm(\l)h(\mu) - \xm(\mu)h(\l)}{\l - \mu} + 2 \xi(1+\mu h(\l))h(\mu)\\
&\quad + 4\xi\l\xm(\mu)\xp(\l) + 2(\xi\mu)^{2}\xp(\mu).
\end{split}
\end{equation}

The next step in the Algebraic Bethe Ansatz is to define appropriate creation operators that yield the Bethe states. 
The creation operators used in the $sl_{2}$-invariant Gaudin model coincide with one of the L-matrix entry \cite{Gaudin2,Sklyanin89}. However, in the present case these operators are non-homogeneous polynomials of the operator $\xm(\l)$. It is convenient to define a more general set of operators in order to simplify the presentation.
\begin{definition}
\label{def:1}
Given integers $M$ and $k\geq 0$, let $\bmu=\left\{\mu_{1},\ldots,\mu_{M}\right\}$ be a set of complex scalars. 
Define the operators
\begin{align}
\label{BMk}
B^{(k)}_{M}(\bmu) &=\left(\xm(\mu_{1})+k\xi\mu_{1}\right)\left(\xm(\mu_{2})+(k+1)\xi\mu_{2}\right)\cdots\left(\xm(\mu_{M}) + (M+k-1)\xi\mu_{M}\right)\notag\\
&= \prod_{\substack{n=k\\ \rightarrow}}^{M+k-1} \left(\xm(\mu_{n-k+1}) + n\,\xi\mu_{n-k+1}\right),
\end{align}
with $B^{(k)}_{0}=1$ and $B^{(k)}_{M}=0$ for $M<0$. 
\end{definition}
The following lemma describes some properties of the $B^{(k)}_{M}(\bmu)$ operators which will be used later on.

\begin{lemma}
\label{lem:1}
Some useful properties of $B^{(k)}_{M}(\mu_{1},\ldots,\mu_{M})$ operators are:
\begin{equation}\begin{split}
\label{recrel}
\mathrm{i)}\quad B^{(k)}_{M}(\mu_{1},\ldots,\mu_{M}) &=B^{(k)}_{m}(\mu_{1},\ldots,\mu_{m})B^{(k+m)}_{M-m}(\mu_{m+1},\ldots,\mu_{M}),\ \ m=1,\ldots,M\\
& = B^{(k)}_{1}(\mu_{1})B^{(k+1)}_{M-1}(\mu_{2},\ldots,\mu_{M})\\
& =  B^{(k)}_{M-1}(\mu_{1},\ldots,\mu_{M-1})B_{1}^{(M+k-1)}(\mu_{M});
\end{split}\end{equation}

$\mathrm{ii)}$ The operators $B^{(k)}_{M}(\mu_{1},\ldots,\mu_{M})$ are symmetric functions of their arguments;\\[0.1cm]

$\mathrm{iii)}\quad B^{(k)}_{M}(\bmu)= B^{(k-1)}_{M}(\bmu) + \xi\sum_{i=1}^{M}\mu_{i}B^{(k)}_{M-1}(\bmu^{(i)})$.
\end{lemma}
The notation used above is the following. Let $\bmu=\left\{\mu_{1},\ldots,\mu_{M}\right\}$ be a set of complex scalars, then
$$
\bmu^{(i_{1},\ldots,i_{k})}=\bmu\setminus\left\{\mu_{i_{1}},\ldots,\mu_{i_{k}}\right\}
$$
for any distinct $i_{1},\ldots,i_{k}\in\left\{1,\ldots,M\right\}$. 

The Algebraic Bethe Ansatz requires the commutation relations between the entries of the L-operator and the $B^{(k)}_{M}(\bmu)$ operators. To this end the notation for the Bethe operators are introduced
\begin{equation}
\label{eq:26}
\betheop_{M} (\l;\bmu) = h(\l) + \sum_{\mu\in\bmu}\frac{2}{\mu - \l},
\end{equation}
where $\bmu=\{\mu_{1},\ldots,\mu_{M-1}\}$ and $\l\in\bbC\setminus\bmu$. In the particular case $M=1$, $\betheop_{1} (\l;\emptyset) = h(\l)$ and will be denoted by $\betheop_{1} (\l)$. The required commutators are given in the following lemma.

\begin{lemma}
\label{lem:2}
The commutation relations between the operators $h(\l)$, $X^{\pm}(\l)$ and the $B^{(k)}_{M}(\mu_{1},\ldots,\mu_{M})$ operators are given by

\begin{align}
\label{eq:2.27}
h(\l)B^{(k)}_{M}(\bmu) &= B^{(k)}_{M}(\bmu) h(\l) + 2\sum_{i=1}^{M}\frac{B^{(k)}_{M}(\l\cup\bmu^{(i)})-B^{(k)}_{M}(\bmu)}{\l-\mu_{i}}\notag\\
&\quad + \xi \sum_{i=1}^{M}B^{(k+1)}_{M-1}(\bmu^{(i)})\bigg(\mu_{i}\betheop_{M}(\mu_{i};\bmu^{(i)})-2k\bigg);
\end{align}

\begin{align}
\label{eq:2.28}
\xp(\l)B^{(k)}_{M}(\bmu) &= B^{(k)}_{M}(\bmu) \xp(\l) - 2\sum_{\substack{i,j=1\\i<j}}^{M}\frac{B^{(k+1)}_{M-1}(\l\cup\bmu^{(i,j)})}{(\l-\mu_{i})(\l-\mu_{j})}\notag\\
&\quad - \sum_{i=1}^{M} B^{(k+1)}_{M-1}(\bmu^{(i)})\left( \frac{\betheop_{M}(\l;\bmu^{(i)})-\betheop_{M}(\mu_{i};\bmu^{(i)})}{\l-\mu_{i}} - \xi \mu_{i}\xp(\l)\right);
\end{align}
\begin{align}
\label{eq:2.29}
\xm(\l)B^{(k)}_{M}(\bmu) &= B^{(k)}_{M}(\bmu)\left(\xm(\l) + M\xi\l \right) -\xi \sum_{i=1}^{M} \mu_{i}B^{(k)}_{M}(\l\cup\bmu^{(i)})\notag\\
&=B^{(k)}_{M+1}(\l\cup\bmu) -\xi \sum_{i=1}^{M} \mu_{i}B^{(k)}_{M}(\l\cup\bmu^{(i)}).
\end{align}
\end{lemma}

It is important to notice that the creation operators that yield the Bethe states of the system are the operators $B^{(0)}_{M}(\bmu)$, below denoted by $B_{M}(\bmu)$. These so-called B-operators are essentially based on those proposed by Kulish in \cite{Kulish2002} and then used in the implementation of the ABA for the Gaudin model based on the r-matrix \eqref{rJordan} 
\cite{CAM1}.  
A recursive relation defining the creation operators follows from i) of the lemma \ref{lem:1}
\begin{equation}
\label{eq:22}
B_{M}(\bmu) = B_{M-1}(\bmu^{(M)})\left(\xm(\mu_{M}) + (M-1)\xi\mu_{M} \right).
\end{equation}
%
As a particular case of lemma \ref{lem:2} the commutation relations between operators $h(\l)$, $X^{\pm}(\l)$ and the B-operators can be obtained by setting $k=0$ in \eqref{eq:2.27}-\eqref{eq:2.29}.

The main step in the implementation of the ABA is to calculate the commutation relations between the creation operators, in this case the B-operators, and the generating function of the integrals of motion $t(\l)$. The complete expression of the required commutator is given in the lemma below and is based on the results established previously in this section. This will be crucial for determining the spectrum of the system and thus constitutes one of the main results of the paper.
\begin{lemma}
\label{lem:4}
The commutation relations between the generating function of the integrals of motion $t(\l)$ and the $B$-operators are given by

\begin{align}
\label{eq:2.34}
t(\l)B_{M}(\bmu) &= B_{M}(\bmu)  \left( t(\l) - \sum_{i=1}^{M} \frac{4h(\l)}{\l - \mu_{i}} + \sum_{i<j}^{M} \frac{8}{\left(\l - \mu_{i}\right) \left(\l - \mu_{j}\right)}+4 M \xi \l\,\xp(\l)\right) \notag\\
& + 4\sum_{i=1}^{M} \frac{B_{M}(\l\cup\bmu^{(i)})}{\l - \mu_{i}}\betheop_{M}(\mu_{i};\bmu^{(i)}) \notag\\
&+2\xi\sum_{i=1}^{M}B^{(1)}_{M-1}(\bmu^{(i)})(\mu_{i}h(\l)+1) \betheop_{M}(\mu_{i};\bmu^{(i)})  \notag \\
%
& + 4\xi\sum_{\substack{i,j=1\\i\neq j}}^{M} \mu_{i}\frac{B^{(1)}_{M-1}(\l\cup\bmu^{(i,j)})-B^{(1)}_{M-1}(\bmu^{(i)})}{\l-\mu_{j}}\betheop_{M}(\mu_{i};\bmu^{(i)})  \notag\\
& + \xi^{2}\sum_{\substack{i,j=1\\i\neq j}}^{M} \mu_{i}B^{(2)}_{M-2}(\bmu^{(i,j)})\left(\mu_{j}\betheop_{M-1}(\mu_{j};\bmu^{(i,j)})-2\right)\betheop_{M}(\mu_{i};\bmu^{(i)})  \notag\\
& + 2\xi^{2}\sum_{i=1}^{M}\mu_{i}^{2} B^{(1)}_{M-1}(\bmu^{(i)}) \xp(\mu_{i}).
\end{align}
\end{lemma}

\begin{proof}
In the case when $M=1$ the commutator between the operators $t(\l)$ and $B_{1}(\mu) = \xm(\mu)$ has been calculated already, see \eqref{eq:16}. 

The next step is to outline the proof in the case when $M > 1$. The starting point in the calculation of the commutator $[t(\l),B_{M}(\bmu)]$ has to be the expression \eqref{t(l)}, so that 
\begin{equation}
\label{tBM-1st}
\begin{split}
\left[t(\l),B_{M}(\bmu)\right] &= \left[h(\l),\left[h(\l),B_{M}(\bmu)\right]\right] + 2\left[h(\l),B_{M}(\bmu)\right] h(\l) - 2 \frac{d}{d\l} \left[h(\l),B_{M}(\bmu)\right]\\
&+ 2\left(2\xm(\l)+\xi\l\right)\left[\xp(\l),B_{M}(\bmu)\right] +4\left[\xm(\l),B_{M}(\bmu)\right]\xp(\l).
\end{split}
\end{equation}
Notice that the terms on the right hand side of the equation above involve only the commutators between the entries of the L-operator and the B-operators. Each of them will be calculated separately, but previously an auxiliary result is established. 
Then the first term on the right hand side of the equation \eqref{tBM-1st} can be calculated using the expression \eqref{eq:2.27} 
\begin{align}
\label{eq:2.37}
&\left[h(\l),\left[h(\l),B_{M}(\bmu)\right]\right] =B_{M}(\bmu)\sum_{i<j}^{M}\frac{8}{(\l-\mu_{i})(\l-\mu_{j})}\notag\\
& - 4\sum_{i=1}^{M}\left(\frac{B_{M}(\l\cup\bmu^{(i)}) - B_{M}(\bmu)}{(\l-\mu_{i})^{2}} - \frac{\frac{d}{d\l}B_{M}(\l\cup\bmu^{(i)})}{\l-\mu_{i}}\right)\notag\\
& + 8\sum_{i < j}^{M}\frac{\xm(\l)B^{(1)}_{M-1}(\l\cup\bmu^{(i,j)})}{(\l-\mu_{i})(\l-\mu_{j})} - 4\sum_{i=1}^{M}\frac{B_{M}(\l\cup\bmu^{(i)})}{\l-\mu_{i}}\sum_{j\neq i}^{M}\frac{2}{\l-\mu_{j}}\notag\\
& + 4\xi\l\sum_{i<j}^{M}\frac{B^{(1)}_{M-1}(\l\cup\bmu^{(i,j)})}{(\l-\mu_{i})(\l-\mu_{j})} + 2\xi\sum_{i=1}^{M} B^{(1)}_{M-1}(\bmu^{(i)})\frac{\l\betheop_{M}(\l;\bmu^{(i)})-\mu_{i}\betheop_{M}(\mu_{i}\bmu^{(i)})}{\l-\mu_{i}}\notag\\
+ \xi\sum_{i\neq j}^{M} &\mu_{i}\left(4\frac{B^{(1)}_{M-1}(\l\cup\bmu^{(i,j)})-B^{(1)}_{M-1}(\bmu^{(i)})}{\l-\mu_{j}} + \xi B^{(2)}_{M-2}(\bmu^{(i,j)})\left(\mu_{j}\betheop_{M-1}(\mu_{j};\bmu^{(i,j)})-2\right)\right)\betheop_{M}(\mu_{i};\bmu^{(i)})\notag\\
& + 2\xi^{2}\sum_{i=1}^{M} \mu_{i}B^{(1)}_{M-1}(\bmu^{(i)})(\l\xp(\l)-\mu_{i}\xp(\mu_{i})).
\end{align}
The equation \eqref{eq:2.27}, with $k=0$, is also used to determine the next two terms on the right hand side of the equation \eqref{tBM-1st}
\begin{equation}
\label{eq:2.38}
\begin{split}
&2\left[h(\l),B_{M}(\bmu)\right] h(\l) = B_{M}(\bmu)\sum_{i=1}^{M}\frac{-4\,h(\l)}{\l-\mu_{i}} + 4\sum_{i=1}^{M}\frac{B_{M}(\{\l\}\cup\bmu^{(i)})}{\l-\mu_{i}}h(\l)\\
& + 2\xi\sum_{i=1}^{M}\mu_{i}B_{M-1}^{(1)}(\bmu^{(i)}) h(\l)\betheop_{M}(\mu_{i};\bmu^{(i)})-4\xi^{2}\sum_{i=1}^{M}\mu_{i}B_{M-1}^{(1)}(\bmu^{(i)})(\l\xp(\l)-\mu_{i}\xp(\mu_{i})),
\end{split}
\end{equation}
and after differentiating \eqref{eq:2.27}, with $k=0$,
\begin{equation}
\label{eq:2.39}
- 2 \frac{d}{d\l} \left[h(\l),B_{M}(\bmu)\right] = 4\sum_{i=1}^{M}\frac{B_{M}(\l\cup\bmu^{(i)}) - B_{M}(\bmu)}{(\l-\mu_{i})^{2}}-\frac{\frac{d}{d\l}B_{M}(\l\cup\bmu^{(i)})}{\l-\mu_{i}}.
\end{equation}
The forth term on the right hand side of the equation \eqref{tBM-1st} can be obtained directly using the expression for the commutator $[\xp(\l),B_{M}(\bmu)]$ given in \eqref{eq:2.28}, with $k=0$,
\begin{align}
& 2\left(2\xm(\l)+\xi\l\right)\left[\xp(\l),B_{M}(\bmu)\right] = -4\sum_{i=1}^{M} B_{M} (\l\cup\bmu^{(i)}) \frac{\betheop_{M}(\l;\bmu^{(i)})-\betheop_{M}(\mu_{i};\bmu^{(i)})}{\l-\mu_{i}}\notag \\
& -8\sum_{i<j}^{M}\frac{\xm(\l)B^{(1)}_{M-1}(\l\cup\bmu^{(i,j)})}{(\l-\mu_{i})(\l-\mu_{j})} - 2 \xi\l\sum_{i=1}^{M} B^{(1)}_{M-1}(\bmu^{(i)})\frac{\betheop_{M}(\l;\bmu^{(i)})-\betheop_{M}(\mu_{i};\bmu^{(i)})}{\l-\mu_{i}} \notag \\
& -4\xi\l\sum_{i<j}^{M}\frac{B^{(1)}_{M-1}(\l\cup\bmu^{(i,j)})}{(\l-\mu_{i})(\l-\mu_{j})} + 4\xi\,\sum_{i=1}^{M} \mu_{i}B_{M}(\l\cup\bmu^{(i)})\xp(\l)\notag \\
&+2\xi^{2}\l\,\sum_{i=1}^{M} \mu_{i}B^{(1)}_{M-1}(\bmu^{(i)})\xp(\l).
\end{align}
Finally, the last term on the right hand side of the equation \eqref{tBM-1st} is determined using the equation \eqref{eq:2.29}, with $k=0$,
\begin{equation}
\label{eq:2.41}
4\left[\xm(\l),B_{M}(\bmu)\right]\xp(\l) = 4\,M\,\xi\l\, B_{M}(\bmu)\xp(\l) - 4\,\xi\sum_{i=1}^{M}\mu_{i}B_{M}(\l\cup\bmu^{(i)})\xp(\l).
\end{equation}

The last step in the proof of the equation \eqref{eq:2.34} is a straightforward addition of all the necessary terms \eqref{eq:2.37}-\eqref{eq:2.41} calculated above.
\end{proof}

Computationally the equation \eqref{eq:2.34} is the crucial step of the ABA. In the next section this result will be used to determine the spectrum of the system. 

At the very end of this section a property of the B-operators relevant for the calculation of the inner products and the norms of the Bethe states is established. Namely, using the realization \eqref{eq:09}, the relation between the B-operators, given by the recurrent relation \eqref{eq:22}, and the corresponding ones for the $sl_2$-invariant Gaudin model, the case when $\xi = 0$, is obtained
\begin{equation}
\label{eq:54}
B_{M}(\bmu) = \sum_{k=0}^{M-1}\xi^{k} \sum_{j_{1<\cdots<j_{M-k}}}^{M} B_{M-k}(\mu_{j_{1}},\ldots,\mu_{j_{M-k}})_{0} \;\pop_{k}^{(M-k)}\left(\bmu^{(j_{1},\ldots,j_{M-k})}\right) + \xi^{M} \pop_{M}^{(0)}(\bmu)
\end{equation}
where 
$$
B_{M-k}(\mu_{j_{1}},\ldots,\mu_{j_{M-k}})_{0}=B_{M-k}(\mu_{j_{1}},\ldots,\mu_{j_{M-k}})|_{\xi=0}
$$
and the operators $\pop_{i}^{(j)}(\mu_{1},\ldots,\mu_{i})$ are given by
\begin{equation}
\label{phat}
\pop_{i}^{(j)}(\mu_{1},\ldots,\mu_{i})=\mu_{1}(-\frac{h_{gl}}{2}+j)\cdots\mu_{i}(-\frac{h_{gl}}{2}+i+j-1) , \quad \text{and} \quad \pop_{0}^{(i)}=1.
\end{equation}

%

In the next section it will be shown how the full implementation of the Algebraic Bethe Ansatz is based on the properties of the creation operators $B_{M}(\bmu)$ obtained here. 

%
%
\section{Spectrum and Bethe vectors of the model}

In this section the spectrum and the Bethe vectors of the Gaudin model based on the r-matrix \eqref{defJordan} 
will be obtained by the ABA. The first step of the ABA is the observation that the  the highest spin vector  $\O_{+}$ \eqref{Omega} is an eigenvector of the generating function of integrals of motion $t(\l)$, due to \eqref{t(l)}, \eqref{vacuum} and \eqref{Ro}. Then Bethe vectors $\Psi_{M}(\bmu)$ are given by the action of the $B$-operators on 
the highest spin vector $\Psi_{M}(\bmu)=B_{M}(\bmu)\O_{+}$. Finally the spectrum of the system is obtained as a consequence of the commutation relations between $t(\l)$ and $B_{M}(\bmu)$ lemma \ref{lem:4}, \eqref{eq:2.34}. The unwanted terms coming from the commutator are annihilated by the Bethe equations on the parameters $\bmu=\{\mu_{1},\ldots,\mu_{M}\}$ as well as the condition $\xp(\l)\O_{+}=0$. Thus the implementation of Algebraic Bethe Ansatz in this case can be resumed in the following theorem.
\begin{theorem}
\label{thm:1}
The highest spin vector $\O_{+}$ \eqref{Omega} is an eigenvector of the operator $t(\l)$ 
\begin{equation}
\label{Lambda0}
t(\l)\O_{+} = \L_{0}(\l) \O_{+}
\end{equation}
with the corresponding eigenvalue 
\begin{equation}
\label{eq:35}
\L_{0}(\l) = \rho^{2}(\l) - 2\rho^{\prime}(\l) 
= \sum_{a=1}^{N}\frac{2}{\l - z_{a}}\left(\sum_{b\neq a}^{N}\frac{\ell_{a}\ell_{b}}{z_{a} - z_{b}}\right)
+ \sum_{a=1}^{N}\frac{\ell_{a}(\ell_{a}+2)}{(\l - z_{a})^{2}} .
\end{equation}
Furthermore, the action of the $B$-operators on the highest spin vector $\O_{+}$ yields the 
Bethe vectors
\begin{equation}
\label{eq:36}
\Psi_{M}(\bmu)=B_{M}(\bmu)\O_{+},
\end{equation}
so that
\begin{equation}
\label{eq:37}
t(\l)\Psi_{M}(\bmu) = \L_{M}(\l;\bmu)\Psi_{M}(\bmu),
\end{equation}
with the eigenvalues
\begin{equation}
\label{GMspc}
\L_{M}(\l;\bmu) = \rho_{M}^{2}(\l;\bmu) - 2\frac{\partial \rho_{M}}{\partial\l}(\l;\bmu)\quad\text{and}\quad\rho_{M}(\l;\bmu)=\rho(\l)-\sum_{i=1}^{M}\frac{2}{\l-\mu_{i}},
\end{equation}
provided that the Bethe equations are imposed on the parameters $\bmu=\left\{\mu_{1},\ldots,\mu_{M}\right\}$
\begin{equation}
\label{eq:39}
\rho_{M}(\mu_{i};\bmu^{(i)}) = \sum_{a=1}^{N}\frac{\ell_{a}}{\mu_{i}-z_{a}} - \sum_{j\neq i}^{M}\frac{2}{\mu_{i}-\mu_{j}} = 0,\quad i=1,\dots,M.
\end{equation}
\end{theorem}
\begin{proof}
The first step in the proof of the Theorem \ref{thm:1} is to show that the highest spin vector  $\O_{+}$ \eqref{Omega} is an eigenvector of the generating function of integrals of motion $t(\l)$. Using \eqref{t(l)}, \eqref{vacuum} and \eqref{t(l)} one gets
\begin{equation}
\label{}
    t(\l)\O_{+} = \left( h^{2}(\l) -2 h^{\prime}(\l) + 2\left(2\xm(\l)+\xi\l\right)\xp(\l) \right) \O_{+} =\L_{0}(\l)\O_{+}\notag
\end{equation}
with the eigenvalue $\L_{0}(\l) =(\rho^{2}(\l) -2 \rho^{\prime}(\l))$ and the function $\rho(\l)$ is given in the equation \eqref{Ro}. 
The next step is an important observation that the action of the operator $t(\l)$ on the Bethe vectors $\Psi_{M}(\bmu)$ \eqref{eq:36} is given by
\begin{equation*}
\label{ }
t(\l)\Psi_{M}(\bmu) = t(\l) B_{M}(\bmu)\O_{+} = \L_{0}(\l)\Psi_{M}(\bmu)+\left[t(\l),B_{M}(\bmu)\right]\O_{+}. 
\end{equation*}
From the equation above it is evident that the main step in the calculation of the action of the generating function of the integrals of motion on the Bethe vectors is computing the commutator between the operator $t(\l)$ and the creation operators $B_{M}(\bmu)$. From the Lemma \ref{lem:4} follows
\begin{align}
\label{tcomPsi}
\!\!\!t(\l)\Psi_{M}(\bmu)\! &= \left( \L_{0}(\l) - \sum_{i=1}^{M} \frac{4\rho(\l)}{\l - \mu_{i}} + \sum_{i<j}^{M} \frac{8}{(\l - \mu_{i})(\l - \mu_{j})}\right)\Psi_{M}(\bmu) \notag\\
&+ 4\sum_{i=1}^{M} \frac{\Psi_{M}(\l\cup\bmu^{(i)})}{\l - \mu_{i}}\rho_{M}(\mu_{i};\bmu^{(i)}) + 2\xi\sum_{i=1}^{M}
\rho_{M}(\mu_{i};\bmu^{(i)}) \times \notag\\
\times &\!\left( (1+\mu_{i}\rho(\l)) \Psi^{(1)}_{M-1}(\bmu^{(i)}) + 2\mu_{i} \sum_{i\neq j}^{M} \frac{\Psi^{(1)}_{M-1}(\l\cup\bmu^{(i,j)})-\Psi^{(1)}_{M-1}(\bmu^{(i)})}{\l-\mu_{j}}\right)\notag\\ 
& + \xi^{2} \sum_{i=1}^{M}\mu_{i}\left(\sum_{j\neq i}^{M} \left(\mu_{j}\rho_{M-1}(\mu_{j};\bmu^{(i,j)})-2\right) \Psi^{(2)}_{M-2}(\bmu^{(i,j)})\right)\rho_{M}(\mu_{i};\bmu^{(i)}).
\end{align}
The unwanted terms in \eqref{tcomPsi} are all proportional to the function $\rho_{M}(\mu_{i};\bmu^{(i)})$ and thus vanish
when the Bethe equations $\rho_{M}(\mu_{i};\bmu^{(i)}) = 0$, \eqref{eq:39} are imposed on the parameters $\bmu=\{\mu_{1},\ldots,\mu_{M}\}$.
Then the equation \eqref{tcomPsi}  has the required form \eqref{eq:37} 
with the eigenvalues given by formula \eqref{GMspc}.
\end{proof}

Although the deformed Gaudin model does not admit the global $sl_{2}$ symmetry, it is of interest to notice that, in the cases when $M$ could be equal to $1 + \textonehalf \sum_{a=1}^N \ell_{a}$, the on-shell Bethe vectors \eqref{eq:36} are highest weight vectors for the action of the global $sl_{2}$ algebra on $\cH$ i.e., 
$$
\xp_{(gl)} \Psi_{M}(\bmu) = 0 \ \text{and} \ h_{(gl)} \Psi_{M}(\bmu) = (-2 M ) \Psi_{M}(\bmu), \ \text{when} \ \, M =1 + \frac{1}{2} \sum_{a=1}^N \ell_{a}
$$
and the Bethe equations \eqref{eq:39} are imposed on the parameters $\bmu=\{\mu_{1},\ldots,\mu_{M}\}$. 
In particular, the formulae \eqref{eq:09} yield
\begin{equation}
\label{globalgen}
\xp_{(gl)} = \lim _{\l \to \infty} \left(\l \xp(\l) \right) \ \text{and} \ \, h_{(gl)} = \left(-\frac{2}{\xi} \right) \lim _{\l \to \infty} \left( \frac{\xm(\l)}{\l}\right) .
\end{equation}
The statement above follows from the formulae \eqref{globalgen}, the lemma \ref{lem:2}, the equations \eqref{eq:2.28} and \eqref{eq:2.29}, and the lemma \ref{lem:1}, the equation \eqref{recrel}.

From the commutator
\begin{align}
\left[t(\l),\xm_{(gl)}\right] &= -4\xi \left(\xm(\l) \hat{X}^{+}_{(gl)}+\l\xm_{(gl)}\xp(\l)\right) + 2\xi h(\l) \left(\hat{h}_{(gl)} - \l h_{(gl)}\right)\notag\\ 
& \quad+ 2\xi h_{(gl)} + \xi^{2}\left(h_{(gl)}+\hat{h}_{(gl)}\right)\hat{X}^{+}_{(gl)} + 2\xi^{2}\hat{X}^{+}_{(gl)},
\end{align}
it follows that,  even if the dimension of the invariant subspace associated with each Bethe vector is the same as for the invariant model, the vectors $(\xm_{(gl)})^{n}\Psi_{M}(\bmu)$ are not eigenvectors of $t(\l)$ and, in general, they even do not generate the corresponding invariant subspace.


To obtain the spectrum of the Gaudin Hamiltonians $H^{(a)}$ \eqref{GHam} one uses the pole expansion \eqref{poleexp} of the generating function $t(\l)$ and the previous theorem.
\begin{corollary}
\label{cor:2}
The Bethe vectors $\Psi_{M}(\bmu)$ \eqref{eq:36} are eigenvectors of the Gaudin Hamiltonians \eqref{GHam}
\begin{equation}
\label{GHsSpectra}
H^{(a)}\Psi_{M}(\bmu) = E_{M}^{(a)}\Psi_{M}(\bmu),
\end{equation}
with the corresponding eigenvalues
\begin{equation}
\label{GMVp}
E_{M}^{(a)}=\sum_{b\neq a}^{N}\frac{\ell_{a}\ell_{b}}{z_{a}-z_{b}}-\sum_{i=1}^{M}\frac{2\ell_{a}}{z_{a}-\mu_{i}}.
\end{equation}
\end{corollary}

Below it is assumed that  the local parameters $z_{a}$ are real numbers, \textit{i.e.} $z_{a}=\bar{z}_{a}$ for all $a=1,\ldots, N$ and the Hilbert structure in the representation space $\cH$ is as in section 2. As expected, the spectra of the Gaudin Hamiltonians $H^{(a)}$ are real, although these operators are not Hermitian, for all $a=1,\ldots, N$,
\begin{equation} 
H^{(a)}\neq (H^{(a)})^{\ast} = \sum_{b\neq a}^{N}\left(\frac{c_2^{\ox}(a,b)}{z_a-z_b} + \xi ( z_b h_a \xm_b - z_a h_b \xm_a)\right).
\end{equation}
Analogously, the generating functions of the integrals of motion $t(\l)$ is also not Hermitian
\begin{equation}
\label{dualt(l)}
t(\l) \neq (t(\l))^{\ast} =  t(\overline{\l})_{0} + 2\xi\left( \hat{X}^{-}_{(gl)} h(\overline{\l})_{0} +\xm_{(gl)} - \overline{\l}\xm(\overline{\l}) h_{(gl)} - \frac{\xi \overline{\l} ^2}{2} h_{(gl)} ^2\right) + \xi ^2 (\hat{X}^{-}_{(gl)})^{2},
\end{equation}
where $\hat{X}^{-}_{(gl)}=\sum_{a=1}^{N}z_{a}\xm_{a}$. 

The problem of correlation functions of spin chains have been extensively studied \cite{KBI} as well as for the $sl_{2}$-invariant Gaudin model \cite{Sklyanin99}.
In the case of the inner product of the off-shell Bethe vectors, for the system under study, a similar approach to that in \cite{CAM1} is developed below.

In order to establish a relation between the inner products of the Bethe vectors \eqref{eq:36} and the corresponding ones in the $sl_2$-invariant case  it is of interest to consider the property \eqref{eq:54} of the B-operators, established at the end of the section 2, in its dual form
\begin{equation}
\label{dualBMs}
B^{*}_{M}(\bl) = \sum_{k=0}^{M-1}\bar{\xi}^{\,k} \sum_{i_{1<\cdots<i_{M-k}}}^{M} \pop_{k}^{(M-k)}\left(\bl^{(i_{1},\ldots,i_{M-k})}\right)\; B^{*}_{M-k}(\l_{i_{1}},\ldots,\l_{i_{M-k}})_{0} + \bar{\xi}^{M}\pop_{M}^{(0)}(\bl)
\end{equation}
where $B^{*}_{M}(\bl) = \left( B_{M}(\overline{\bl}) \right)^{*}$ and
\begin{equation}
B^{*}_{N}(\l_{1},\ldots,\l_{N})_{0}=B^{*}_{N}(\l_{1},\ldots,\l_{N})|_{\xi=0}=\xp(\l_{1})\xp(\l_{2})\cdots\xp(\l_{N}).
\end{equation}

Notice that the operators \eqref{dualBMs} do not annihilate the highest spin vector $\O_{+}$ \eqref{Omega}
\begin{equation}
\label{eq:b01}
B^{*}_{M}(\l_{1},\ldots,\l_{M})\O_{+}=\bar{\xi}^{M}\pop_{M}^{(0)}(\bl)\O_{+} = \bar{\xi}^{\,M} p_{M}(\l_{1},\ldots,\l_{M})\O_{+},
\end{equation}
here $p_{M}(\l_{1},\ldots,\l_{M})=\l_{1}\cdots\l_{M}\, p_{M}$ with
\begin{equation}
\label{eq:4.11}
p_{M}=\left(-\frac{\rho_{(gl)}}{2}\right)\left(-\frac{\rho_{(gl)}}{2}+1\right)\cdots\left(-\frac{\rho_{(gl)}}{2}+M-1\right)\quad \text{and}\quad \rho_{(gl)}=\sum_{a=1}^{N}\ell_{a}.
\end{equation}

The following lemma yields an explicit relation between the inner products of the Bethe vectors, when the Bethe equations are not imposed on their parameters ("off-shell Bethe states"), of the deformed Gaudin model and the corresponding ones in the $sl_{2}$-invariant case. 
\begin{lemma}
\label{lem:6}
Let $M_{1},M_{2}$ be two natural numbers and $\bl=\{\l_{1},\ldots,\l_{M_{1}}\}$, $\bmu=\{\mu_{1},\ldots,\mu_{M_{2}}\}$ be sets of complex scalars.
Then the inner product between the off-shell Bethe states \eqref{eq:36} is given by
\begin{equation}
\label{InnerPr}
\left\langle\Psi_{M_{1}}(\bl)\lvert\Psi_{M_{2}}(\bmu)\right\rangle =  \bar{\xi}^{M_{1}}\xi^{M_{2}}\,p_{M_{1}}(\overline{\bl}) p_{M_{2}}(\bmu) + \sum_{k_{1}=0}^{M_{1}-1}\sum_{k_{2}=0}^{M_{2}-1}\bar{\xi}^{\,k_{1}}\xi^{\,k_{2}}  \big\langle \phi_{M_{1},\,k_{1}}(\bl)\lvert\phi_{M_{2},\,k_{2}}(\bmu)\big\rangle,
\end{equation}
where
\begin{equation}
\label{innerP1}
\phi_{M,\,k}(\bl)=\sum_{1\leqslant i_{1}<\cdots<i_{M-k}}^{M} p^{(M-k)}_{k}\left(\bl^{(i_{1},\ldots,i_{M-k})}\right)\Psi_{M-k}\left(\l_{i_{1}},\ldots,\l_{i_{M-k}}\right)_{0},
\end{equation}
with $\Psi_{M}(\cdot)_{0}=\Psi_{M}(\cdot)|_{\xi=0}$ and
\begin{equation}
\begin{split}
\label{phatOm}
\pop_{i}^{(j)}(\l_{1},\ldots,\l_{i}) \O_{+} &= p_{i}^{(j)}(\l_{1},\ldots,\l_{i}) \O_{+} \\
&= \left( \l_{1}(-\frac{\rho_{gl}}{2}+j)\cdots\l_{i}(-\frac{\rho_{gl}}{2}+i+j-1) \right) \O_{+}.
\end{split}
\end{equation}
\end{lemma}
\begin{proof}
Let $\Psi_{M_{1}}(\bl)$ and $\Psi_{M_{2}}(\bmu)$ be two off-shell Bethe states. Then,
\begin{equation}
\label{innerP2}
\langle\Psi_{M_{1}}(\bl)\lvert\Psi_{M_{2}}(\bmu)\rangle = \langle B_{M_{1}}(\bl)\O_{+}\,\lvert\, B_{M_{2}}(\bmu)\O_{+}\rangle= \langle \O_{+}\,\lvert\, B^{*}_{M_{1}}(\overline{\bl})B_{M_{2}}(\bmu)\O_{+}\rangle .
\end{equation}
From the eqns. \eqref{eq:54} and \eqref{dualBMs} it follows
\begin{equation}\begin{split}
\label{DualBxB}
B^{*}_{M_{1}}(\overline{\bl})B_{M_{2}}(\bmu)&= \bar{\xi}^{M_{1}}\, \pop_{M_{1}}^{(0)}(\overline{\bl}) B_{M_{2}}(\bmu) + \xi^{M_{2}}\, B^{*}_{M_{1}}(\overline{\bl})\pop_{M_{2}}^{(0)}(\bmu)
- \bar{\xi}^{M_{1}}\xi^{M_{2}}\, \pop_{M_{1}}^{(0)}(\overline{\bl}) \pop_{M_{2}}^{(0)}(\bmu)\\[0.3cm]
& \!\!\!\!\!\!\!\!\!\!\!\!\!\!\!\!\!\!\!\!\!\!\!\!\!\!\!\!\!\! + \sum_{k_{1}=0}^{M_{1}-1}\sum_{k_{2}=0}^{M_{2}-1}\bar{\xi}^{\,k_{1}}\xi^{\,k_{2}} \; \sum_{i_{1}<\cdots<i_{M_{1}-k_{1}}}^{M_{1}} \sum_{j_{1}<\cdots<j_{M_{2}-k_{2}}}^{M_{2}} \pop^{(M_{1}-k_{1})}_{k_{1}}\left(\overline{\bl}^{(i_{1},\ldots,i_{M_{1}-k_{1})})}\right) \times\\[0.3cm]
&\!\!\!\!\!\!\!\!\!\!\!\!\!\!\!\!\!\!\!\!\!\!\!\!\!\!\!\!\!\!\!\!\!\!\!\!\!\!\!\!\!\!\!\!\times B^{*}_{M_{1}-k_{1}}(\overline{\l}_{i_{1}},\ldots,\overline{\l}_{i_{M_{1}-k_{1}}})_{0}\,
B_{M_{2}-k_{2}}(\mu_{j_{1}},\ldots,\mu_{j_{M_{2}-k_{2}}})_{0}\,\pop^{(M_{2}-k_{2})}_{k_{2}}\left(\bmu^{(j_{1},\ldots,j_{M_{2}-k_{2}})}\right).
\end{split}\end{equation}
Substituing \eqref{DualBxB} into \eqref{innerP2}, and taking into account \eqref{eq:b01}

\begin{equation}\begin{split}
\label{innerP3}
\left\langle\Psi_{M_{1}}(\bl)\lvert\Psi_{M_{2}}(\bmu)\right\rangle &=  \bar{\xi}^{M_{1}}\xi^{M_{2}}\,p_{M_{1}}(\overline{\bl}) p_{M_{2}}(\bmu) + \sum_{k_{1}=0}^{M_{1}-1}\sum_{k_{2}=0}^{M_{2}-1}\bar{\xi}^{\,k_{1}}\xi^{\,k_{2}} \; \times \\
 &\!\!\!\!\!\!\!\!\!\!\!\!\!\!\!\!\!\!\!\!\!\!\!\!\!\!\!\!\!\!\!\!\!\!\!\!\!\!\!\!\!\!\!\!\times \sum_{i_{1}<\cdots<i_{M_{1}-k_{1}}}^{M_{1}} \sum_{j_{1}<\cdots<j_{M_{2}-k_{2}}}^{M_{2}} 
 p^{(M_{1}-k_{1})}_{k_{1}}\left(\overline{\bl}^{(i_{1},\ldots,i_{M_{1}-k_{1})})}\right)\,p^{(M_{2}-k_{2})}_{k_{2}}\left(\bmu^{(j_{1},\ldots,j_{M_{2}-k_{2}})}\right)\; \times\\[0.3cm]
 &\times  \bigg\langle \Psi_{M_{1}-k_{1}}(\l_{i_{1}},\ldots,\l_{i_{M_{1}-k_{1}}})_{0}\bigg\lvert\Psi_{M_{2}-k_{2}}(\mu_{j_{1}},\ldots,\mu_{j_{M_{2}-k_{2}}})_{0}\bigg\rangle.
\end{split}\end{equation}
Given that $\langle\cdot\lvert\cdot \rangle$ is the  Hermitian inner product in the representation space $\cH$, as considered in section 2,
the right hand side of \eqref{innerP3} can be expressed as in \eqref{InnerPr}.
\end{proof}

The main result of this section is the equation \eqref{InnerPr} of the lemma \ref{lem:6} where the inner products
of the off-shell Bethe states of the deformed Gaudin model are expressed in terms of the $sl_{2}$-invariant inner products calculated by Sklyanin \cite{Sklyanin99}. A substantial advantage of this approach is that the result was obtained without doing the cumbersome calculation involving the commutators \eqref{loopalg} between the entries of the L-operator. 

Given that $\langle \phi_{M_{1},\,k_{1}}(\bl)\lvert\phi_{M_{2},\,k_{2}}(\bmu)\big\rangle = 0$ unless $M_1-k_1= M_2-k_2$, the lemma \ref{lem:6} yields an expression of the norm of the Bethe vectors.
\begin{corollary}
\label{norms}
The norms of the off-shell Bethe states $\Psi_{M}(\bmu)$ are given by
\begin{equation}
\label{NormBV}
\left\|\Psi_{M}(\bmu)\right\|^{2} =  \lvert\xi\lvert^{2M}\,\lvert\bmu\lvert^{2}(p_{M})^{2} + \sum_{k=0}^{M-1}\lvert\xi\lvert^{2k}  \left\|\phi_{M,k}(\bmu)\right\|^{2}
\end{equation}
where $\lvert\bmu\lvert^{2}=\lvert\mu_{1}\lvert^{2}\cdots\lvert\mu_{M}\lvert^{2}$ and $\phi_{M,k}(\bmu)$ are defined in \eqref{innerP1}.
\end{corollary}
The corollary above reveals the full strength of our approach in the fact that the expression of the norm of the off-shell Bethe states again involves only the $sl_{2}$-invariant inner products given in \cite{Sklyanin99}. 

\section{Conclusions} 

The Gaudin model based on the  r-matrix \eqref{defJordan} is studied.
The usual Gaudin realization of the model is introduced and the B-operators $B^{(k)}_{M}(\mu_{1},\ldots,\mu_{M})$ are defined as non-homogeneous polynomials of the operator $\xm(\l)$. These operators are symmetric functions of their arguments and they satisfy certain recursive relations with explicit dependency on the quasi-momenta $\mu_{1},\ldots,\mu_{M}$. The creation operators $B_{M}(\mu_{1},\ldots,\mu_{M})$ which yield the Bethe vectors 
form their proper subset. The commutator of the creation operators with the generating function of the Gaudin model under study is calculated explicitly. 
Based on the previous result the spectrum of the system is determined. It turns out 
that the spectrum of the system and the corresponding Bethe equations coincide with the ones of the $sl_{2}$-invariant model. However, contrary to the $sl_{2}$-invariant case, the generating function of integrals of motion and the corresponding Gaudin Hamiltonians are not Hermitian. The explicit form of the generalized Bethe vectors associated to the Jordan canonical form of the generating function $t(\l)$ remains an open problem.

In order to obtain expressions for the inner product of the Bethe vectors the operators $B^{\ast}_{M}(\bmu)$ are introduced. The method used in the study of the inner product of the Bethe vectors is based on the relation between the creation operators of the system and the ones in the $sl_2$-invariant case. The advantage of this approach is that the cumbersome calculation involving strings of commutators between the entries of the L-operator becomes unnecessary. The final expression of the inner products and norms of the Bethe states is given in terms of the corresponding ones of the $sl_{2}$-invariant Gaudin model. As opposed to the $sl_{2}$-invariant case, the Bethe vectors $\Psi_{M}(\bmu)$ for the deformed Gaudin model are not orthogonal for different $M$'s.

\section{Acknowledgments}
We acknowledge stimulating discussions with Professor P. P. Kulish and  remarks 
of the referee. This work was supported by the FCT projects PTDC/MAT/69635/2006 and PTDC/MAT/99880/2008.

\begin{appendix}
\appendixpage
\section{Proofs of lemmas}

The proofs of the lemmas, not included in the main text, are detailed in this appendix.\\

\noindent\textbf{Proof of the lemma \ref{lem:1}}\\
\noindent$\mathrm{i)}$ The recurrence relations \eqref{recrel} are evident from the expression \eqref{BMk} in the definition \ref{def:1}.\\ 
$\mathrm{ii)}$ Given the natural number $k$, and according to the definition \eqref{BMk} with $M=2$
\begin{equation}
\label{app1}
B^{(k)}_{2}(\mu_{1},\mu_{2}) = \xm(\mu_{1})\xm(\mu_{2}) + (k+1)\,\xi\mu_{2} \xm(\mu_{1}) + k\,\xi\mu_{1} \xm(\mu_{2}) + k(k+1)\xi^{2}\mu_{1}\mu_{2}.
\end{equation}
Using \eqref{loopalg}, it is straightforward to verify that $B^{(k)}_{2}(\mu_{1},\mu_{2}) = B^{(k)}_{2}(\mu_{2},\mu_{1})$.
Moreover, it is clear from \eqref{recrel} that
\begin{equation*}\begin{split}
\label{ }
B^{(k)}_{M}(\mu_1,\ldots,\mu_{i},\mu_{i+1},\ldots,\mu_{M})&=B^{(k)}_{i-1}(\mu_{1},\ldots,\mu_{i-1})B^{(k+i-1)}_{2}(\mu_{i},\mu_{i+1})B^{(k+i+1)}_{M-i-1}(\mu_{i+2},\ldots,\mu_{M})\\
&=B^{(k)}_{M}(\mu_1,\ldots,\mu_{i+1},\mu_{i},\ldots,\mu_{M}),
\end{split}\end{equation*}
therefore $B^{(k)}_{M}(\mu_1,\ldots,\mu_{M})$ is a symmetric function.

\noindent$\mathrm{iii)}$ Given the natural number $k$, a direct consequence of the expression \eqref{app1} is that
$$
B^{(k)}_{2}(\mu_{1},\mu_{2}) = B^{(k-1)}_{2}(\mu_{1},\mu_{2}) +\xi (\mu_{2}B^{(k)}_{1}(\mu_{1}) + \mu_{1}B^{(k)}_{1}(\mu_{2})).
$$
If the expression $\mathrm{iii)}$ of the lemma \ref{lem:1} is assumed to be true for some $M\geq 2$, we have, using \eqref{recrel}, 
\begin{align*}
&B^{(k)}_{M+1}(\bmu) = \left(B^{(k-1)}_{M}(\bmu^{(M+1)}) + \xi\sum_{i=1}^{M}\mu_{i}B^{(k)}_{M-1}(\bmu^{(i,M+1)})\right)\left(B^{(M+k-1)}_{1}(\mu_{M+1}) + \xi\mu_{M+1}\right) \\
&=B^{(k-1)}_{M+1}(\bmu) + \xi\mu_{M+1}\, B^{(k-1)}_{M}(\bmu^{(M+1)}) + \xi\,\sum_{i=1}^{M}\mu_{i}\left(B^{(k)}_{M}(\bmu^{(i)}) + \xi\mu_{M+1}\, B^{(k)}_{M-1}(\bmu^{(i,M+1)})\right)\\
&=B^{(k-1)}_{M+1}(\bmu) + \xi\,\sum_{i=1}^{M}\mu_{i}B^{(k)}_{M}(\bmu^{(i)}) + \xi\mu_{M+1} B^{(k)}_{M}(\bmu^{(M+1)}) = B^{(k-1)}_{M+1}(\bmu) + \xi\sum_{i=1}^{M+1}\mu_{i}B^{(k)}_{M}(\bmu^{(i)}).
\end{align*}

\noindent\textbf{Proof of the formula \eqref{eq:2.27} of the lemma \ref{lem:2}}.\\
Let  $k$ be a natural number. From the relations in \eqref{loopalg} and the definition \ref{def:1} it is immediately clear that 
\begin{equation}
\label{app2}
\left[h(\l),B^{(k)}_{1}(\mu)\right] = 2\frac{B^{(k)}_{1}(\l) - B^{(k)}_{1}(\mu)}{\l - \mu} + \xi\left( \betheop_{1}(\mu)-2k\right),
\end{equation}
Assume that the expression \eqref{eq:2.27} of the lemma \ref{lem:2} is true for some $M\geq 1$ then
\begin{align*}
&\left[h(\l),B^{(k)}_{M+1}(\bmu)\right] = \left[h(\l),B^{(k)}_{M}(\bmu^{(M+1)})\right]B_{1}^{(M+k)}(\mu_{M+1}) + B^{(k)}_{M}(\bmu^{(M+1)})[h(\l),B_{1}^{(M+k)}(\mu_{M+1})]\\
& = \sum_{i=1}^{M}\left(2\frac{B^{(k)}_{M+1}(\l\cup\bmu^{(i)})-B^{(k)}_{M+1}(\bmu)}{\l-\mu_{i}} +\xi B^{(k+1)}_{M}(\bmu^{(i)})\left(\mu_{i}\betheop_{M}(\mu_{i};\bmu^{(i,M+1)})-2k\right)\right)\\
& + \xi \sum_{i=1}^{M} \mu_{i} B^{(k+1)}_{M-1}(\bmu^{(i,M+1)})[h(\mu_{i}),B_{1}^{(M+k)}(\mu_{M+1})] + B^{(k)}_{M}(\bmu^{(M+1)})[h(\l),B_{1}^{(M+k)}(\mu_{M+1})]\\
& = 2\sum_{i=1}^{M+1}\frac{B^{(k)}_{M+1}(\l\cup\bmu^{(i)})-B^{(k)}_{M+1}(\bmu)}{\l-\mu_{i}} +\xi \sum_{i=1}^{M} B^{(k+1)}_{M}(\bmu^{(i)})\left(\mu_{i}\betheop_{M}(\mu_{i};\bmu^{(i,M+1)})+\frac{2\mu_{i}}{\mu_{M+1}-\mu_{i}}-2k\right)\\
&+ \xi B^{(k+1)}_{M}(\bmu^{(M+1)})\left(-2k+\sum_{i=1}^{M} \frac{2\mu_{M+1}}{\mu_{i}-\mu_{M+1}}\right) +\xi\left(B^{(k)}_{M}(\bmu^{(M+1)})+\xi\sum_{i=1}^{M}\mu_{i}B^{(k+1)}_{M-1}(\bmu^{(i,M+1)})\right)h(\mu_{M+1})\\
&= \sum_{i=1}^{M+1}\left(2\frac{B^{(k)}_{M+1}(\l\cup\bmu^{(i)})-B^{(k)}_{M+1}(\bmu)}{\l-\mu_{i}} + \xi B^{(k+1)}_{M}(\bmu^{(i)})\left(\mu_{i}\betheop_{M+1}(\mu_{i};\bmu^{(i)})-2k\right)\right),
\end{align*}
where the equation \eqref{app2} and the relation \eqref{recrel} of the lemma \ref{lem:1}, are used when appropriate.
\end{appendix}

The proofs of the formulas \eqref{eq:2.28} and \eqref{eq:2.29} are analogous to the demonstration above.

\end{document}